\documentclass[a4paper,11pt]{article}

\usepackage{graphicx}
\usepackage{amssymb}
\usepackage{amsmath}
\usepackage{amsfonts}
\usepackage{amsthm}
\usepackage{epstopdf}
\usepackage{algorithm}
\usepackage[noend]{algpseudocode}
\usepackage{algorithmicx}
\usepackage[utf8]{inputenc}
\usepackage{todonotes}
\usepackage{authblk}
\usepackage[OT4]{fontenc}
\usepackage{fullpage}

\usepackage{url}
\urldef{\mailsa}\path<{a.karczmarz|j.lacki}@mimuw.edu.pl<

\newcommand{\fa}{\mathtt{forall}}
\newcommand{\fea}{\mathtt{forall2}}
\newcommand{\ex}{\mathtt{exists}}
\newcommand{\lar}{\mathtt{l}}
\newcommand{\rar}{\mathtt{r}}
\newcommand{\repr}{\mathtt{repr}}
\newcommand{\scut}{\mathtt{shortcut}}

\newcommand{\compid}{\textsc{Comp-Id}}
\newcommand{\reduc}{\textsc{Reduce}}
\newcommand{\contr}{\textsc{Contract}}
\newcommand{\ctree}{\textsc{Compute-Tree}}
\newcommand{\faint}{\textsc{Forall-Aux}}
\newcommand{\reptr}{\textsc{Report-Triangles}}
\newcommand{\maybeqed}{}
\newcommand{\lson}{\text{left}}
\newcommand{\rson}{\text{right}}
\newcommand{\imid}{\text{mid}}
\newcommand{\ipar}{\text{par}}

\makeatletter
\newcommand{\newreptheorem}[2]{%
\newenvironment{rep#1}[1]{%
\spnewtheorem*{rep@theorem#1}{#2 \ref{##1}}{\normalfont\bfseries}{\itshape}
 \def\rep@title{#2 \ref{##1}}%
 \begin{rep@theorem#1}}%
 {\end{rep@theorem#1}}}
\makeatother

\newtheorem{theorem}{Theorem}
\newtheorem{lemma}{Lemma}

\newtheorem{corollary}{Corollary}

\begin{document}

\title{Fast and simple connectivity in graph timelines}
\date{}
\author{Adam Karczmarz 
\thanks{Supported by the grant NCN2014/13/B/ST6/01811 of the Polish Science Center. Partially supported by FET IP project MULTIPLEX 317532.}}
\author{Jakub Łącki
\thanks{Jakub Łącki is a recipient of the Google Europe Fellowship in Graph Algorithms, and this research is supported in part by this Google Fellowship.}}
\affil{University of Warsaw}
\affil{\mailsa}

\maketitle

\begin{abstract}
In this paper we study the problem of answering connectivity queries
about a \emph{graph timeline}.
A graph timeline is a sequence of undirected graphs $G_1,\ldots,G_t$ on
a common set of vertices of size $n$ such that each graph is obtained
from the previous one by an addition or a deletion of a single edge.
We present data structures, which preprocess the timeline
and can answer the following queries:
\begin{itemize}
  \item $\fa(u,v,a,b)$ -- does the path $u\to v$ exist in \emph{each} of $G_a,\ldots,G_b$?
  \item $\ex(u,v,a,b)$ -- does the path $u\to v$ exist in \emph{any} of $G_a,\ldots,G_b$?
  \item $\fea(u,v,a,b)$ -- do there exist two edge-disjoint paths connecting $u$ and $v$ in \emph{each} of $G_a,\ldots,G_b$?
\end{itemize}
We show data structures that can answer $\fa$ and $\fea$ queries
in $O(\log n)$ time after preprocessing in $O(m+t\log n)$ time.
Here by $m$ we denote the number of edges that remain unchanged in each graph of the timeline.
For the case of $\ex$ queries, we show how to extend an existing data structure
to obtain a preprocessing/query trade-off of $\langle O(m+\min(nt, t^{2-\alpha})), O(t^\alpha)\rangle$
and show a matching conditional lower bound.
\end{abstract}

\section{Introduction}

In this paper we revisit the problem of maintaining the connectivity
information in a \emph{graph timeline}.
The problem was formulated and solved in a recent paper by Łącki and Sankowski
\cite{Lacki:2013}.
They define a graph timeline to be a sequence of graphs
$G_1,G_2,\ldots,G_t$ on a common set of vertices $V$ of size $n$ such that
the graph $G_i$ is obtained from $G_{i-1}$ by adding or deleting a single edge.
Their goal was to preprocess the graph timeline to build a data structure that
may answer connectivity queries regarding a contiguous fragment of the timeline:
\begin{itemize}
\item $\fa(u, v, a, b)$ --- are vertices $u$ and $v$ connected by a path in \emph{each} of $G_a, G_{a+1}, \ldots, G_b$?
\item $\ex(u, v, a, b)$ --- are vertices $u$ and $v$ connected by a path in \emph{any} of $G_a, G_{a+1}, \ldots, G_b$?
\end{itemize}
We stress that the entire timeline is revealed in the very beginning for preprocessing,
and after that the queries may arrive in an online fashion.

Throughout this paper, we write $\langle f(n, m, t), g(n, m, t)\rangle$ to denote a data structure,
whose preprocessing time is $f(n, m, t)$ and the query time is $g(n, m, t)$.

In the case of $\fa$ queries, Łącki and Sankowski presented an
$\langle O(m+t\log t\log \log t\log n),$ $O(\log n\log \log t)\rangle$ data structure.
Here by $m$ we denote the number of edges that remain unchanged in each of $G_1,\ldots,G_t$.
Their data structure is Monte Carlo randomized and the query time is amortized.
For $\ex$ queries they give an $\langle O(m+nt),O(1)\rangle$ data structure.

We improve the results of~\cite{Lacki:2013} and show new algorithms, which are more efficient, simpler and deterministic.
In addition, we also develop an extended data structure that may efficiently answer an even more complex query regarding 2-edge-connectivity:
\begin{itemize}
\item $\fea(u, w, a, b)$ --- are vertices $u$ and $v$ connected by two edge-disjoint paths in \emph{each} of $G_a, G_{a+1}, \ldots, G_b$?
\end{itemize}
Moreover, we give new conditional lower bounds for the problem of answering $\ex$ queries, which also improves the results
of~\cite{Lacki:2013}.

\subsection{Related work}
A rich body of connectivity-related dynamic problems has been studied in the area of networks and distributed computing.
A number of such problems has been surveyed in~\cite{Casteigts12}.
In a typical scenario, we work with a sequence of graphs $G^t = G_1, \ldots, G_t$ that represent the states of an evolving network at different points in time.
However, the properties of these graphs, which are of interest, such as \emph{T-interval connectivity}~\cite{Kuhn10} or \emph{time-respecting paths}~\cite{Kempe02} are usually much more complex than what can be studied with ordinary connectivity queries, that is queries about the existence of a path connecting two given vertices in a particular graph.
For example, the problem of T-interval connectivity consists of deciding if for every subsequence $G_a, \ldots, G_{a+T-1}$ of $T$ consecutive graphs in~$G^t$, the intersection $G_a \cap \ldots \cap G_{a+T-1}$ of these graphs contains a~connected component spanning all vertices.
Here we define the intersection of two graphs to be the graph obtained by intersecting their edge sets.

We believe that the queries we consider in this paper are powerful enough to study interesting properties of evolving networks.
A $\fa$ query checks if two vertices are connected with a~path in every graph among $G_a, \ldots, G_b$, but the path can be different in each of the graphs and may not even exist in the intersection of these graphs.
Even stronger is a $\fea$ query, checking whether two vertices are connected with two edge-disjoint paths in each graph of the given fragment.
This may serve as a measure of robustness of connection between two nodes of a~network.

The algorithms that process graph timelines can also be considered \emph{semi-offline} counterparts of dynamic graph algorithms.
The updates are given upfront, but the queries may arrive in an online fashion,
i.e. they are issued one by one, only after the preprocessing is finished.
A~possible scenario for the semi-offline model would be to collect and index the history of evolving network up to
some point of time and then use the queries to analyze
various properties of the network efficiently.

It is worth noting that the knowledge of the entire history of changes in most cases leads to
data structures faster and simpler than the best online ones.
However, this property has rarely been exploited to design efficient algorithms.
Eppstein~\cite{Eppstein:1994} has shown an algorithm, which, given a weighted graph $G$ and a sequence of $k$
edge weight updates, computes the weight of the minimum spanning tree after each update in $O((m + k)\log n)$ time.

\subsection{Our results}
We show $\langle O(m + t \log n), O(\log n)\rangle$ data structures for answering $\fa$ queries and $\fea$ queries.
The data structures use $O(t\log n)$ space.
This improves the results of~\cite{Lacki:2013} in a number of ways: our algorithms are faster and deterministic, use less space,
the time bounds are worst-case and the query time is independent of the length of the timeline.
We also introduce $\fea$ queries, which were not considered before.
On top of that, our algorithms are arguably simpler.

What is interesting, we obtain a solution for the 2-edge-connectivity problem, which is much more efficient than what has been achieved in the dynamic case.
The best known algorithm for 2-edge-connectivity is due to Holm et al.~\cite{Holm:2001}.
It processes $t$ updates in $O((t+m) \log^4 n)$ time, where $m$ is the initial number of edges, and answers queries in $O(\log n)$ time.
Our algorithm may preprocess the timeline in only $O(m + t \log n)$ time to answer queries in $O(\log n)$ time.

In the construction of the algorithm for answering $\fa$ queries we use the following two observations.
Consider a timeline $G_1, \ldots, G_t$.
If there is an edge~$uw$ present in every graph among $G_1, \ldots, G_t$,
vertices $u$ and $w$ are equivalent from the point of view of any query, so the edge $uw$ can be contracted in each graph.
Once we do that, we are left with $O(t)$ edges in total, each being added or deleted at some point of time.
Thus, if there are much more than~$t$ vertices, some vertices are isolated in every $G_1, \ldots, G_t$, and can be safely treated separately in the beginning and removed.
These ideas are then used recursively in a divide-and-conquer algorithm, which at each step halves the length of the timeline to compute a segment tree over the sequence $G_1, \ldots, G_t$.
This segment tree stores connectivity information about every individual graph in the timeline.
Here we adapt the ideas of Eppstein's reduction and contraction scheme used for offline computation of minimum spanning trees~\cite{Eppstein:1994}.

Next, we use a fingerprinting scheme to identify vertices belonging to the same connected components in multiple consecutive graphs,
which allows us to answer $\fa$ queries.
Additionally, our fast algorithm for answering queries uses a data structure for efficient testing of equality of contiguous subsequences of a given sequence.
This is then extended to handle $\fea$ queries.

For $\ex$ queries, we show how to leverage the
$\langle O(m+nt), O(1)\rangle$ data structure from~\cite{Lacki:2013}
to build an $\langle O(m+\min(nt,t^{2-\alpha})), O(t^\alpha)\rangle$ data structure,
where $\alpha$ is a parameter from the range $[0,1)$, which can be chosen arbitrarily.
All of the presented algorithms are simple and can easily be implemented.

Moreover, we develop a conditional lower bound for the problem of answering $\ex$ queries.
We show that answering $t$ $\ex$ queries on a timeline of length~$t$, consisting of graphs with $O(t)$ edges,
can be used to detect triangles in a graph with $O(t)$ edges.
This implies a conditional lower bound of $\Omega(t^{1.41})$ and improves the result of~\cite{Lacki:2013}, where a weaker lower bound was shown.
We also show that an $O(t^{1.5-\epsilon})$ \emph{combinatorial} algorithm for the aforementioned problem would imply a subcubic \emph{combinatorial} algorithm for
the Boolean matrix multiplication problem, which would be a major breakthrough.
At the same time, our improved data structure for $\ex$ queries may solve this problem in $O(t^{1.5})$ time, which means that it is, in some sense, optimal.

\subsection{Organization of this paper}
In Section~\ref{sec:preliminaries} we introduce notation and give a few simple properties of segment trees, which we later use.
Section~\ref{sec:tree_structure} describes the basic version of our data structure, which is then extended to handle $\fa$ and $\fea$ queries.
Then, in Section~\ref{sec:forall} we present an algorithm for answering $\fa$ queries.
Next, in Section~\ref{sec:exists} we develop improved lower bounds for the problem of answering $\ex$ queries, as well as show that a trade-off between query and preprocessing time is possible.
Finally, in Section~\ref{sec:open_problems} we discuss the possible directions of future research.

\section{Preliminaries}\label{sec:preliminaries}

A \emph{graph timeline} is a sequence $G^t$ of graphs $G_1,G_2,\ldots,G_t$,
where $G_i=(V,E_i)$.
We call each individual graph in $G^t$ a \emph{version}.
For each $i\in[1,t)$ we have ${|E_i\oplus E_{i+1}|=1}$, i.e. $E_{i+1}$
is obtained from $E_i$ by adding or deleting a single edge.
We assume that the input is given as the set $E_1$ and
a list of $t-1$ operations that describe, for each $i \in [1, t-1]$, how to obtain $E_{i+1}$ from $E_i$.

Throughout this paper we work with intervals of integers, that is $[a, b]$ denotes
\linebreak $\{a, {a+1}, \ldots, b\}$.
We say that edge $(u,v)$ is \emph{alive} in the interval $[x,y]$ iff $(u,v)\in E_j$
for each $j\in [x,y]$.
For each edge $e \in E_1 \cup \ldots \cup E_t$ we define $L(e)$ to be the set of maximal
intervals such that $e$ is alive in each of them.
An edge $e$ is called \emph{permanent} iff $L(e)=\{[1,t]\}$, that is,
it is present in every version.
Otherwise, we say that $e$ is a \emph{temporary} edge.
We denote by $m$ the number of permanent edges.
The number of temporary edges is at most $t$.
We begin the initialization of our data structures by finding
the sets $L(e)$ in $O(|E_1|+t)=O(m+t)$ time.

We denote by $\Delta_a^+$ the set of edges $e$ such that $[a,x]\in L(e)$
for some ${x\in[a,t]}$, i.e., edges present in $G_a$, but not in $G_{a-1}$.
Similarly, let $\Delta_b^-$ be the set of edges~$e$ such that $[x,b]\in L(e)$
for some $x\in[1,b]$.
It is easy to verify that $\sum_{i=1}^t |\Delta_i^+|+\sum_{i=1}^t|\Delta_i^-|=O(m+t)$.
Moreover, for $a\in(1,t]$, we have $|\Delta_a^+|\leq 1$, while
for $b\in[1,t)$ we have $|\Delta_b^-|\leq 1$.

Throughout the paper, we assume that $t\geq n$ and $t=2^B$ for some integer $B\geq 0$.
The latter assumption can be achieved by adding dummy graphs to the timeline.
\subsection{Elementary intervals and the segment tree}
Given $t=2^B$, the set of \emph{elementary intervals} is defined
inductively:
\begin{enumerate}
\item $[1,t]$ is an elementary interval,
\item if $[a,b]$ is an elementary interval, and $a<b$ we let
  $\imid = \left\lfloor \frac{a+b}{2}\right\rfloor$, and define $[a,\imid]$ and
  $[\imid+1,b]$ to be elementary intervals as well.
\end{enumerate}
The set of elementary intervals can be naturally organized into a complete binary tree,
which we call a \emph{segment tree}.
Assuming the above notation, we call $\lson([a,b])=[a,\imid]$
the left child of interval $[a,b]$.
Similarly, $\rson([a,b])=[\imid+1,b]$.
The parent interval of $P$ is denoted by $\ipar(P)$.
We first prove a few properties of elementary intervals.

\begin{lemma}\label{l:anc}
  Every two elementary intervals are either disjoint, or
    one of them is contained in the other.
    The latter is the case iff one of them is a descendant of
    the other in the segment tree.
\end{lemma}

\begin{lemma}\label{l:elem}
Every interval $[c,d] \subseteq [1,t]$
can be partitioned into no more than $2\log_2(d-c+1)+2$ disjoint elementary
intervals such that no two intervals from the partition
can be merged into a bigger elementary interval.
The partition can be computed in time $O(\log(d-c+1))$.
\end{lemma}
\begin{proof}
If $c=d$, then the interval does not have to be partitioned at all.
Assume $c<d$.
Consider the leaves $[c,c]$ and $[d,d]$ of the segment tree and let $P$ be
the lowest common ancestor of these intervals, i.e., the smallest elementary interval
which contains both $c$ and $d$.
Our initial partition is formed by the following intervals:
\begin{itemize}
\item $[c,c]$ and $[d,d]$,
\item if both the interval $Q$ and its parent lie on the path from
  $[c,c]$ to $P$ (but excluding $P$) and also $Q=\lson(\ipar(Q))$, we include $\rson(\ipar(Q))$
  (i.e. the sibling of $Q$) in our partition,
\item if both the interval $Q$ and its parent lie on the path from
  $[d,d]$ to $P$ (but excluding $P$) and also $Q=\rson(\ipar(Q))$, we include $\lson(\ipar(Q))$
  (i.e. the sibling of $Q$) in our partition.
\end{itemize}
We first show that the chosen family of intervals $W$ is indeed
a partition of $[c,d]$.
By Lemma~\ref{l:anc}, the chosen intervals are disjoint, since there
are no two such that one of them is an ancestor of the other.
For any interval from $W$, its left endpoint is not less than $c$,
whereas its right endpoint is not larger than $d$.
Hence, $\bigcup W\subseteq [c,d]$.
Moreover, $[c,d]\subseteq \bigcup W$.
It is clear that $\{c, d\} \subseteq \bigcup W$.
To show that $f \in (c,d)$ belongs to $\bigcup W$
consider a path from $[f,f]$ to $P$.
This path either joins the path $[c,c]\to P$
from the right, or joins the path $[d,d]\to P$ from the left.
In both of these cases, the last interval $Q$ of $[f,f]\to P$
before the paths merged ($f\in Q\subseteq [c,d]$) was included in $W$.

Let us count the number of intervals in $W$.
First notice, that every elementary interval in~$W$ is not longer
than $d-c+1$.
Furthermore, each subsequent interval chosen from one of the paths
($[c,c]\to P$ or $[d,d]\to P$) is at least twice as long as the
previous interval taken while climbing that path.
Taking into account the additional intervals $[c,c]$ and $[d,d]$,
we get the bound $2\log_2 (d-c+1)+2$.
The $O(\log (d-c+1))$ time can be achieved by climbing
the two paths simultaneously.

The above procedure does not guarantee that no two elementary
intervals from $W$ can be merged into a larger elementary interval.
However, this can be easily fixed.
Every time when we put into $W$ an elementary interval
such that its sibling in the segment tree is already contained
in $W$, we replace the two siblings with their parent.
As the lengths of the elementary intervals put into $W$
only increase on a path $[c,c]\to P$ or $[d,d]\to P$,
the potential sibling can only be the interval that was the last
to be included in $W$.

Eventually, we might also end up with $W=\{\lson(P), \rson(P)\}$;
then we ought to replace the partition with $\{P\}$.

This fix does not influence the overall time complexity of the
partitioning, which remains $O(\log (d-c+1))$.
\maybeqed\end{proof}

\begin{lemma}\label{l:disj}
If $P_1,P_2,\ldots,P_k$ are disjoint intervals contained in $[1,t]$,
we can partition them into at most $2k\left(\log_2{\frac{t}{k}}+1\right)$
disjoint elementary intervals.
\end{lemma}

\begin{proof}
We use Lemma~\ref{l:elem} to partition each of $P_1, \ldots, P_k$.
For $1 \leq i \leq k$, let $l_i = |P_i|$.
Since the intervals are disjoint, their partitions into elementary intervals
are also disjoint.
Hence, by Lemma~\ref{l:elem}, the total size of the partition can be
bounded as follows:
$$2k+2\sum_{i=1}^k\log_2{l_i}\\
\leq 2k+2k\log_2\left(\frac{1}{k}\sum_{i=1}^k l_i\right )\leq 2k+2k\log_2{\frac{t}{k}} \\
  =2k\left(\log_2{\frac{t}{k}}+1\right).$$

  We used the bound $\sum_{i=1}^k l_i\leq t$ and the Jensen's inequality
  for the concave function \linebreak $f(x)=\log_2{x}$.
\maybeqed\end{proof}

As it is much easier to work with elementary intervals,
for each edge $e$ we partition all intervals from $L(e)$
into elementary intervals.
\begin{lemma}\label{l:part}
  All intervals in $\bigcup_{e\in V\times V}L(e)$
  can be partitioned
  into $O(m+t\log n)$ elementary intervals.
  The partition can be performed in time $O(m+t\log n)$.
\end{lemma}

\begin{proof}
Denote by $E^*$ the set of temporary edges.
For any $e\in E^*$, let us denote by $q_e$ the number $|L(e)|$.
We have $\sum_{e\in E^*}q_e \in [\frac{t}{2},t]$ and $|E^*|\leq \min(t,n^2)$.
By Lemma \ref{l:disj}, we conclude that the total number of elementary
intervals for temporary edges is at most
\begin{align*}
2\sum_{e\in E^*}\left(q_e\log_2\left(\frac{t}{q_e}\right)+1\right)&  \leq 2|E^*|\left(\sum_{e\in E^*}\frac{q_e}{|E^*|}\log_2\left(\frac{t}{q_e}\right)\right)+2t\\
 \leq&2|E^*|\left(\frac{1}{|E^*|}\sum_{e\in E^*} q_e\right) \log_2\left(\frac{t}{\left(\frac{1}{|E^*|}\sum_{e\in E^*} q_e\right)}\right) +2t\\
     \leq &2t\log_2\left(2|E^*|\right)+2t\\
    =&O(t\log n).
\end{align*}

Here we used the Jensen's inequality for the concave function $f(x)=x\log_2{\frac{t}{x}}$
and weights equal to $\frac{1}{|E^*|}$.
Since each permanent edge has exactly one interval in its partition,
we obtain the desired bound $O(m+t\log{n})$.

\maybeqed\end{proof}

For an elementary interval $[a,b]$, we set $E_{[a,b]}$ to be the set
of edges that contain $[a,b]$ in their partition.
From Lemmas~\ref{l:elem} and \ref{l:part} it follows that each edge
is contained in $O(\log{t})$ sets $E_{[a,b]}$ and the
sum over elementary intervals
$\sum_{[a,b]} E_{[a,b]}$ is of order $O(m+t\log{n})$.


\section{The data structure}\label{sec:tree_structure}
We now describe a tree-like data structure $T$, which is a crucial part of all our algorithms.
In the following we reserve the name $T$ for this particular data structure.
The data structure~$T$ is based on the set of all elementary intervals organized into a complete binary tree.
This tree has a single node $T_{[a,b]}$ for each elementary interval $[a,b]$.
Denote by $G_{[a,b]}$ the graph $(V,E_a\cap \ldots \cap E_b)$.
Roughly speaking, our goal is to associate with $T_{[a,b]}$ the information about the connected components of $G_{[a,b]}$.
We first give a simple approach for constructing the data structure $T$,
and then show how to speed it up.
We use the following fact.

\begin{lemma}\label{l:egab}
Let $[a,b]$ be an elementary interval such that $[a,b] \neq [1,t]$.
Then \linebreak $E(G_{[a,b]}) = E(G_{\ipar([a,b])}) \cup E_{[a,b]}$.
\end{lemma}
\begin{proof}
Recall that $E(G_{[a,b]}) = E(G_a) \cap \ldots \cap E(G_b)$.
Thus, $E(G_{\ipar([a,b])}) \subseteq E(G_{[a,b]})$.
Moreover, directly from the definitions we have  $E_{[a,b]} \subseteq E(G_{[a,b]})$.
It remains to show $E(G_{[a,b]}) \subseteq E(G_{\ipar([a,b])}) \cup E_{[a,b]}$.

Consider an edge $e \in E(G_{[a,b]}$.
If $e$ is alive in some interval $[c,d] \supseteq [a,b]$ such that $\ipar([a,b])\subseteq [c,d]$,
we have $e \in E(G_{\ipar([a,b])})$.
To complete the proof it remains to consider the case when
$e$ is alive in an interval $[c,d]$, and $\ipar([a,b])\not\subseteq [c,d]$.
We show that in this case $e \in E_{[a,b]}$.
In other words, $[a,b] \in X$, where $X$ is the partition of elements
of $L(e)$ into elementary intervals.

From $\ipar([a,b])\not\subseteq [c,d]$ it follows that no ancestor
of $[a,b]$ is a part of $X$.
At the same time, every elementary interval that is neither an
ancestor nor a descendant of $[a,b]$ is disjoint with $[a,b]$ (Lemma~\ref{l:anc}), 
so it does not belong to~$X$ either.
Therefore, each elementary interval that belongs to $X$ and intersects $[a,b]$
is either $[a,b]$ or one of its descendants.
Consequently, $X$ contains a subset of disjoint elementary intervals $Y$
such that $\bigcup Y=[a,b]$.
Suppose $Y\neq \{[a,b]\}$. 
Let~$Q_1$ be the shortest interval from $Y$. Note that $Q_1\neq [a,b]$.
Also, by Lemma~\ref{l:elem}, the sibling~$Q_2$ of $Q_1$ is not
contained in $Y$.
As $Q_1$ is the shortest, no descendant of $Q_2$ is contained in $Y$.
Moreover, as $Q_1\in Y$, no ancestor of $Q_2$ is contained in $Y$.
Hence, $Y$ does not cover the integers from $Q_2$, a contradiction.
Thus, we have $Y=\{[a,b]\}$, so $[a,b]\in X$ and, as a result, we have $e\in E_{[a,b]}$.
\maybeqed\end{proof}


In the simple approach, we associate with $T_{[a,b]}$ a graph $S_{[a,b]}$, which has a single vertex for each connected component of $G_{[a,b]}$,
and does not contain any edges.
By Lemma~\ref{l:egab}, $G_{[a,b]}$ is obtained from $G_{\ipar([a,b])}$ by adding some edges.
This implies that each component of $G_{[a,b]}$ is a sum of some components of $G_{\ipar([a,b])}$.
To compute $S_{[a,b]}$ we build a graph $H$ on a vertex set $V(S_{\ipar([a,b])})$ and add to it edges of $E_{[a,b]}$ (each edge endpoint has to be mapped to its connected component in $G_{\ipar([a,b])}$) and then find its connected components.
These components are exactly the components of $G_{[a,b]}$.
Observe that during this computation we may also compute a mapping between the vertices of
$S_{\ipar([a,b])}$ and $S_{[a,b]}$.
In the case of $S_{[1,t]}$ we compute a mapping between individual vertices and connected components of $G_{[1,t]}$.

$T$ represents the connected components of every graph in the timeline.
Consider a graph $G_c$.
In order to find a connected component of a vertex $v$ in $G_c$,
we traverse the path in $T$ from $T_{[1,t]}$ to $T_{[c,c]}$.
We compute the connected component of vertex $v$
in every graph $G_{[a,b]}$ on the path.
Observe that if we know the connected component of $v$ in $G_{\ipar([a,b])}$,
we may compute the connected component of $v$ in $G_{[a,b]}$ by
following the mapping between the components of $G_{\ipar([a,b])}$ and $G_{[a,b]}$.
At the end of the traversal, we find the component of
$v$ in $G_{[c,c]} = G_c$.

\subsection{An efficient construction}\label{s:access}
In order to compute the data structure $T$ efficiently, we need to make an additional
optimization, which is crucial for obtaining good running time.

Consider an elementary interval $[a,b]$ and a connected component $C$
of $G_{[a,b]}$.
Assume that within the graphs $G_a, \ldots, G_b$ no edge incident to a
vertex of $C$ is ever added or deleted.
In other words, the edges incident to vertices of $C$ are the same in each of
$G_a, \ldots, G_b$.
This means that in each of $G_a, \ldots, G_b$ vertices of $C$ are connected to
each other, but not connected to \emph{any} vertex outside $C$.
Hence,~$C$ is also a connected component in each of $G_a, \ldots, G_b$.

As a result, there is no need to store $C$ in the descendants of $T_{[a,b]}$.
When searching for a connected component of a vertex $v \in C$ in $G_c$, where $c \in [a,b]$,
we may simply stop the search in the representation of $C$ in $T_{[a,b]}$.
This observation will be used in the reduction phase of the construction of the tree $T$.

We now describe the efficient construction of the tree $T$.
For each node $T_{[a,b]}$ of $T$, where $[a,b]$ is an elementary interval, we compute a graph $S_{[a,b]}$.
The vertices of $S_{[a,b]}$ correspond to \emph{some} of the components of $G_{[a,b]}$.
We say that $v\in V$ is \emph{represented} in $S_{[a,b]}$
if there is a vertex $s\in V(S_{[a,b]})$ that corresponds to a component containing $v$.
The graphs $S_{[a,b]}$ have no edges.\footnote{Defining a graph with no edges may look confusing.
However, we define $S_{[a,b]}$ to be a graph, as we add edges to $S_{[a,b]}$ in our data structure
for 2-edge-connectivity.}

Let $[a,b]$ be an elementary interval.
$S_{[a,b]}$ is computed based
on $S_{\ipar([a,b])}$ (or $(V,\emptyset)$, if $[a,b] = [1,t]$)
in two phases called \emph{reduction} and \emph{contraction}.

In the reduction phase some vertices of $H=S_{\ipar([a,b])}$ are removed, as they are not affected by
any edge addition or deletion that is carried out among $G_a, \ldots, G_b$.
Namely, we mark endpoints of edges in $F = E_{[a,b]} \cup \bigcup_{i=a+1}^b \Delta_i^{+} \cup \bigcup_{i=a}^{b-1} \Delta_i^{-}$
and then remove the unmarked vertices.
Note that the sets $E_{[a,b]}, \Delta_i^{+}$ and $\Delta_i^{-}$ contain edges
of the original graph, so their endpoints have to be mapped to the corresponding vertices of $H$.
The reduction phase is performed only when $b - a + 1 < n$.
It is done by a call $\reduc(H,F)$, which produces
a pair $(S',M)$, where $S'$ is the reduced graph and $M$ is a mapping between $V(S_{\ipar([a,b])})$ and $V(S')\cup \{\perp\}$.
The value of~$\perp$ means that a vertex has been removed and does not have a corresponding vertex in $S'$.
The procedure can be implemented with a simple graph search to work in $O(|H| + |F|)$ time.

In the second phase, called the contraction phase, some of the remaining vertices of $H=S'$ are merged to form $S_{[a,b]}$.
Specifically, the components formed in $S'$ after adding edges $F=E_{[a,b]}$ are contracted.
Again, we use a function $\contr(H,F)$, which produces a pair $(S', M)$ consisting
of the contracted graph $S'$ and the mapping between $H$ and $S'$.
This function can also be easily implemented to work in linear time.

Consider an elementary interval $P$.
Together with $S_P$, the node $T_P$
stores two tables $\lar_P$ and $\rar_P$
mapping vertices of $S_P$ to $V(S_{\lson(P)})\cup\{\perp\}$
and $V(S_{\rson(P)})\cup\{\perp\}$ respectively.
If $\lar_P[k]\neq \perp$, $\lar_P[k]$ is the vertex of $S_{\lson(P)}$ that
corresponds to $k\in V(S_P)$.
$\lar_P[k]=\perp$ means that $P$ is a leaf, or
there is no vertex corresponding to $k$ in $S_{\lson(P)}$.
The table $\rar_P$ is defined analogously.
For simplicity, we also assume that $T_{[1,t]}$ is a left
child of a special node $T_{[0,\infty]}$ and
$S_{[0,\infty]}=(V,\emptyset)$, so that for each
$v\in V$, $\lar_{[0,\infty]}[v]$ points
to the vertex of $S_{[1,t]}$ representing the original vertex $v$.

The graphs $S_P$ along with $\lar$ and $\rar$ pointers are
sufficient to find the component of any vertex $v$ in any of $G_1, \ldots, G_t$.
To access the component of vertex $v$ in $G_c$ we start at vertex $v$ in $S_{[0,\infty]}$ and follow $\lar$ or $\rar$ pointers
in order to reach the leaf $T_{[c,c]}$.
The traversal stops once we reach~$T_{[c,c]}$ or the pointer we want to use
($\lar[k]$ or $\rar[k]$) is equal to $\perp$.
Let $P$ be the elementary interval, where the traversal finishes and $k$ be the vertex in $S_P$, which we reached.
Then, as we later show, $(k,P)$ uniquely identifies the component of vertex $v$ in $G_c$.
The above process can be seen as a function
$\compid(w,a,b,c)$ that follows the path to $T_{[c,c]}$
starting at vertex $w\in V(S_{[a,b]})$.
The pair $(k,P)$, defined as above, is what the call
$\compid(\lar_{[0,\infty]}[v],1,t,c)$ returns.
The full text of the $\compid$ function is given
in Appendix \ref{a:code}.

\begin{lemma}\label{lem:comp-id}
Let $1 \leq c \leq t$.
For any $u \in V$, denote by $(k_u,P_u)$ the value returned by \linebreak
$\compid(\lar_{[0,\infty]}[v],1,t,c)$.
Then, two vertices $v,w\in V$ are connected by a path in $G_c$ iff $k_v=k_w$ and $P_v=P_w$.
\end{lemma}
\begin{proof}
Observe that $[c,c] \subseteq P_v$ and $[c,c] \subseteq P_w$.
First, assume that $P_v\neq P_w$, which, by Lemma~\ref{l:anc}
means that one of the intervals contains the other one.
Without loss of generality suppose that $P_v\subseteq \lson(P_w)$.
Then $k_w$ is a component of $G_{P_w}$ that is not incident to any changes
in the time interval $\lson(P_w)$, while $v$ is in some component of $G_{P_w}$
that undergoes changes in $\lson(P_w)$.
Thus, these are different components.
If $P_v=P_w$, then both $k_v$ and $k_w$ are components of $G_{P_v}$ not
incident to any changes in the time interval $P_v$.
Both $v$ and $w$, however, are represented in $S_{P_v}$, so they are in
the same component iff $k_v=k_w$.
\maybeqed\end{proof}

Let us bound the time needed to build $T$.
We begin with an auxiliary lemma,
whose proof is based on the fact that we perform the reduction.

\begin{lemma}\label{lem:short_sp}
  Let $[a,b]$ be an elementary interval.
  Then $|V(S_{[a,b]})|\leq \min(8(b-a+1),n)$.
\end{lemma}
\begin{proof}
Clearly,  $|V(S_{[a,b]})|\leq n$.
For $b-a+1\geq n$, we have $8(b-a+1)\geq n$, so $|V(S_{[a,b]})|\leq n\leq\min(n,8(b-a+1))$
holds.
Assume that $b-a+1<n$.
  The graph $S_{[a,b]}$ is constructed from $S_{\ipar([a,b])}$ by applying
reduction and contraction.
Let $C=\bigcup_{i=a+1}^{b}\Delta_i^+\cup\bigcup_{i=a}^{b-1}\Delta_i^-$.
The reduction produces a graph $S'$ of at most $2|E_{[a,b]}\cup C|$
vertices.
The contraction does not increase the number of vertices.
Therefore, $|V(S_{[a,b]})|\leq 2|E_{[a,b]}\cup C|$.

Let $\ipar([a,b])=[a_1,b_1]$ and consider the analogous
set $C_1$ for $\ipar([a,b])$, i.e. \linebreak
$C_1=\bigcup_{i=a_1+1}^{b_1}\Delta_i^+\cup\bigcup_{i=a_1}^{b_1-1}\Delta_i^-$.
For $1 < i \leq t$ we have $|\Delta_i^+|\leq 1$ and for $1 \leq i < t$ we have $|\Delta_i^-|\leq 1$.
Moreover, since $\ipar([a,b])=[a_1,b_1]$, we have $b_1 - a_1 + 1 = 2(b-a+1)$.
Thus, $|C_1| \leq 2(b_1 - a_1+1) = 4(b-a+1)$.
To complete the proof, we show that both $C$ and $E_{[a,b]}$ are
subsets of $C_1$, which implies that $|V(S_{[a,b]})| \leq 2|E_{[a,b]}\cup C| \leq 2|C_1| \leq 8(b-a+1)$.
Clearly, $C\subseteq C_1$, as the sum in $C$ goes through
less summands than the sum defining $C_1$.
To show that $E_{[a,b]} \subseteq C_1$, consider $e\in E_{[a,b]}$.
Suppose that $[a,b]$ is the left child of $\ipar([a,b])=[a_1,b_1]$.
We show that $e\in\bigcup_{i=a_1}^{b_1-1}\Delta_i^-$.
By Lemma $\ref{l:elem}$, we have both $e\notin E_{[a_1,b_1]}$
and $e\notin E_{\rson([a_1,b_1])}$.
Thus, the edge $e$ is deleted in some version~$G_j$ for $j\in[b+1,b_1]$,
which means $e\in \Delta_{j-1}^-$.
Analogously we prove that if $[a,b]=\rson([a_1,b_1])$
then $e\in\bigcup_{i=a_1+1}^{b_1}\Delta_i^+$.
Hence $e\in C_1$.
\maybeqed\end{proof}

To build $T$ we use a recursive procedure $\ctree(a,b)$, which computes the subtree rooted at $T_{[a,b]}$.
It produces each graph $S_{[a,b]}$ based on $S_{\ipar([a,b])}$ by applying reduction and contraction.
During the computation of $T$,
we maintain an auxiliary table $\repr$, fulfilling
the following invariant: both at the beginning and at the end of the call
$\ctree(a,b)$, $\repr[v]$ is the vertex of $S_{\ipar([a,b])}$
representing $v\in V$, if such vertex exists.
Initially, we have $\repr[v]=v$, which does not break the
invariant, as we have previously set
$\ipar([1,t])=[0,\infty]$ and $S_{[0,\infty]}=(V,\emptyset)$.
The $\repr$ table is used implicitly by the procedures $\reduc$
and $\contr$ to map the endpoints of edges from
 $E_{[a,b]}$, $\Delta_i^+$ and $\Delta_i^-$ to vertices
of $S_{\ipar([a,b])}$ in constant time.

All the computed tables use linear space and can be
accessed in constant time, as we can identify the vertices
of introduced graphs $S_P$ with natural numbers $\{1,2,\ldots\}$
and the $\perp$ value with $0$.
The total used space is asymptotically no more than the time
spent on computing $T$, that is $O(m+t\log{n})$.
\begin{algorithm}[H]
\begin{algorithmic}[1]
    \Procedure{compute-tree}{$a, b$} \Comment{$[a,b]$ -- elementary interval}
        \State $P:=\ipar([a,b])$
        \If{$b-a+1 < n$} \Comment{Reduction is only done for short elementary intervals.}
            \State $C:=\bigcup_{i=a+1}^b \Delta_i^+ \cup \bigcup_{i=a}^{b-1} \Delta_i^-$
            \State $U := $vertices of $V$ incident with any edge of $E_{[a,b]}\cup C$
            \State $(S',M'):=\Call{reduce}{S_P, E_{[a,b]}\cup C}$
        \Else
            \State $U := V$
            \State $(S',M') = (S_P, \texttt{id})$
        \EndIf
        \For{$u\in U$}
          \State $\texttt{mem}[u]=\repr[u]$ \Comment{remember old $\repr$ values}
        \EndFor
        \For{$u \in U$}
            \State $\repr[u]$ := $M'(\repr[u])$
        \EndFor
        \State $(S'',M''):=\Call{Contract}{S', E_{[a,b]}}$ \Comment{$M''$ maps $V(S_P)$ to $V(S_{[a,b]})$.}
        \State $S_{[a,b]}:=S''$
        \For{$k \in S_{[a,b]}$} \Comment{initialize $\lar$ and $\rar$ pointers to $\perp$}
          \State $\lar_{[a,b]}[k]:=\perp$
          \State $\rar_{[a,b]}[k]:=\perp$
        \EndFor
        \For{$u \in U$}
            \State $\repr[u]$ := $M''(\repr[u])$
        \EndFor
        \For{$s \in S_P$} \Comment{set the parent $\lar$ and $\rar$ pointers}
            \State $s':=M'(s)$
            \If{$s'\neq \perp$}
                \State $s':=M''(s')$
            \EndIf
            \If{$[a,b]=\lson(P)$}
                \State $\lar_P[s]:=s'$
            \Else
                \State $\rar_P[s]:=s'$
            \EndIf
        \EndFor
        \If{$a<b$} \Comment{compute the children}
            \State $\imid:=\left\lfloor\frac{a+b}{2}\right\rfloor$
            \State \Call{compute-tree}{$a, \imid$}
            \State \Call{compute-tree}{$\imid+1, b$}
        \EndIf
        \For{$u\in U$} \Comment{restore \texttt{repr} to the initial state}
            \State $\repr[u]:=\texttt{mem}[u]$
        \EndFor
    \EndProcedure
\end{algorithmic}
\end{algorithm}

\begin{lemma}
The total running time of $\ctree(1,t)$ is $O(m+t\log n)$.
\end{lemma}

\begin{proof}
We first analyze the time spent in the call $\ctree(a,b)$,
excluding the work in recursive calls.
Let $C$ be $\bigcup_{i=a+1}^{b}\Delta_i^+\cup\bigcup_{i=a}^{b-1}\Delta_i^-$.
Thus $O(|C|) = O(b-a)$.
Recall that the functions $\contr$ and $\reduc$ run in linear time.
For $b-a+1\geq n$, we only perform contraction of $E_{[a,b]}$ in a graph of size $O(n)$,
which requires $O(n+|E_{[a,b]}|)$ time.
The amount of work for $b-a+1<n$ can be bounded by
$O(|V(S_{\ipar([a,b])})|+|C|+|E_{[a,b]}|)$, as $\reduc$ is passed the edges $C\cup E_{[a,b]}$.

To complete the proof, we sum these running times over all elementary intervals.
The term $|E_{a,b}|$ appears in both cases and,
by Lemma~\ref{l:part}, we have $\sum_P E_P=O(m+t\log n)$,
thus we can focus on the other summands.
For the case $b-a+1\geq n$, the remaining work is $O(n)$,
but there are only $O(\frac{t}{n})$ such intervals, so the total
work is $O(t)$.
On the other hand, if $b - a + 1 < n$, by Lemma~\ref{lem:short_sp}, $O(|V(S_{\ipar([a,b])})|) = O(b-a)$, so the total work is
$O(b - a)$.
Hence, the total work on each level of the tree such that
its elementary intervals are shorter than~$n$, is $O(t)$.
The number of such levels is $O(\log n)$, which gives $O(t \log n)$ total time.
The lemma follows.
\maybeqed\end{proof}

Having computed $T$, the function $\compid$
allows us to access the component of some vertex $v$ in $G_c$
in time $O(\log t)$.
However, as we now show, this can be speeded up to $O(\log n)$ time.
Recall that $t = 2^B$.
Let $2^D$ be the smallest power of $2$ such that $2^D\geq n$
and fix some $k\in[0,2^{B-D})$.
Then, for each $c\in[k\cdot2^D+1,(k+1)\cdot 2^D]$,
the call $\compid(\lar_{[0,\infty]}[v],1,t,c)$ descends
down $T$ through the first $B-D$ levels
in the same way, independent of $c$.
We can thus add another preprocessing phase, building
the table $\scut$.
For a vertex $v$ and $0 \leq k < 2^{B-D}$, $\scut[v][k]$ is defined to be a pair $(s,P)$ such that
for $c \in[k\cdot2^D+1,(k+1)\cdot 2^D]$, $\compid(\lar_{[0,\infty]}[v], 1, t, c)$, after going through at most $B-D$ levels of $T$,
ends up in the interval $P$ and $s\in V(S_P)$ represents $v$.
There are only $O(t/n)$ allowed values of $k$,
so the table $\scut$ has size $O(t)$.

The table can be computed by finding the components
of each vertex $v$ in all the graphs $S_P$ from
the first $B-D$ levels of the tree.
As the component of $v$ in~$S_P$ can be computed in
constant time based on the component of $v$ in $S_{\ipar(P)}$,
we spend $O(t/n)$ time for each $v$, and thus $O(t)$ time in total.

The optimized procedure $\compid$ starts by looking
up the shortcut through first $B-D$ levels of
$T$ and then calls the original $\compid$, starting at
an elementary interval of length $O(n)$.
Thus, its running time is $O(\log n)$.

\subsection{2-edge-connectivity}
As in the case of connectivity, we first show how to preprocess the graph in order to efficiently answer 2-edge-connectivity queries regarding individual versions.
Our approach is similar to the idea of Section~\ref{s:access}:
we construct a data structure~$T$ containing graphs $S_{[a,b]}$,
where $[a,b]$ is an elementary interval.
Note that in the case of connectivity, the graphs $S_{[a,b]}$ do not contain any edges.

First, observe that contracting 2-edge-connected components yields a forest.

\begin{lemma}
Let $W$ be the set of 2-edge-connected components of some graph $G$.
Define the graph $H=(W,F)$, where
$$F=\{(w_1,w_2) : (u,v)\in E(G), u\in w_1, v\in w_2, w_1\neq w_2\}.$$
Then, $H$ is a forest.
\end{lemma}
\begin{proof}
Indeed, if there was a cycle $w_1w_2\ldots w_kw_1$ in $H$, then the
components $w_1,\ldots,w_k$ would form a single 2-edge-connected component.
\maybeqed\end{proof}

In the case of 2-edge-connectivity, the graphs $S_{[a,b]}$ are forests of rooted trees, whose vertices
represent some of the 2-edge-connected components of $G_{[a,b]}$.
Each rooted tree in the forest represents a part of some (ordinary) connected
component incident to some edges alive in the time interval $[a,b]$.

The vertices of $S_{[a,b]}$ are partitioned into two categories.
A vertex $s \in V(S_{[a,b]})$ is a \emph{simple vertex} if and only if it represents a
single 2-edge-connected component of $G_{[a,b]}$.
Otherwise, $s$ is called a \emph{path vertex} and it represents
$k$ ($k\geq 2$) 2-edge-connected components of $G_{[a,b]}$ ---
$c_1,\ldots,c_k$ --- that form a ``path'', i.e. for each
$i\in[1,k)$ there is a single edge in $G_P$ connecting some
vertex of $c_i$ and some vertex of $c_{i+1}$.
We maintain the following invariants.
\begin{enumerate}
\item If $s$ is a root of its tree, or its degree in $S_P$ is other than $2$,
then it is a simple vertex.
\item A path vertex is never adjacent to another path vertex.
\end{enumerate}
In particular, each path vertex is of degree $2$ in $S_{[a,b]}$.
Let $s$ be a path vertex representing a path $c_1,c_2,\ldots,c_k$ of
2-edge-connected components of $G_{[a,b]}$.
If $s_1$ is a parent of $s$ and $s_2$ is a child of $s$,
then the edge $(s_1,s)\in E(S_{[a,b]})$ is actually an edge
between components $s_1$ and $c_1$, while the edge $(s,s_2)\in E(S_{[a,b]})$
actually means $(c_k,s_2)$.

The components $c_1,c_2,\ldots,c_k$ of a path vertex $s\in V(S_{[a,b]})$
have the following property: for each $x\in [a,b]$, either
$c_1,c_2,\ldots,c_k$ are actual 2-edge-connected components
of $G_x$ or they are all parts of a single larger 2-edge-connected
component of $G_x$.
Thus, the path vertex $s$ allows us to trace components $c_i$
in the descendants of $T_{[a,b]}$ in a uniform way.
This in turn will later allow us to keep the graphs $S_{[a,b]}$ small.

Recall that in the case of connectivity, the call $\compid(\lar_{[0,\infty]}[v],1,t,x)$
returns a pair $(s,Q)$ ($s\in V(S_Q)$), where $Q$ is the last
interval on the path $[1,t]\to[x,x]$ in the segment tree
such that $v$ is represented with $s$ in $V(S_Q)$.
The introduction of path vertices forces us to modify what $\compid$
does,
as we need to distinguish between distinct 2-edge-connected components
that are represented by the same path vertex.
Now for $x\in[a,b]$ and $k\in V(S_{[a,b]})$, we set $\compid(k,a,b,x)$
to be the pair $(s,Q)$, where $s$ is a representation of $k$ in~$S_Q$ and
$Q$ is the last interval on the path $[a,b]\to[x,x]$
such that $s$ is a simple vertex.
To access the 2-edge-connected component of $v\in V$ in version $G_x$,
we use the same call $\compid(\lar_{[0,\infty]}[v],1,t,x)$.

The tables $\lar_{[a,b]}$ and $\rar_{[a,b]}$ are defined analogously:
for $s\in V(S_{[a,b]})$, if $\lar_{[a,b]}[s]=\perp$, then $s$ is not represented
in $S_{\lson([a,b])}$.
Otherwise, $\lar_{[a,b]}[s]$ is a vertex of $S_{\lson([a,b])}$ representing $s$.
The table $\rar_{[a,b]}$ is defined analogously.

Having defined graphs $S_{[a,b]}$, we now describe how they can be computed.
As previously, we compute $S_{[a,b]}$ based on $S_{\ipar([a,b])}$, by performing
first the reduction and then the contraction.
The reduction is again performed only if $b-a+1<n$.

The reduction proceeds in phases.
The initial phases involve marking some nodes of $S_{\ipar([a,b])}$,
whereas the latter phases reduce the graph's size.
The path vertices never get marked; they can be instead
merged with other path and simple vertices, forming
``longer'' path vertices of $S_{[a,b]}$.

Let $C$ be again $\bigcup_{i=a+1}^b \Delta_i^+ \cup \bigcup_{i=a}^{b-1} \Delta_i^-$.
In the first phase we mark all the vertices of $S_{\ipar([a,b])}$
incident to edges in $E_{[a,b]}\cup C$.
It might be the case that for some original edge $(u,v)$,
$u$ and $v$ are already represented by the same $s$ in $S_{\ipar([a,b])}$ ---
then we just skip this edge.
As a result, no more than $2|E_{[a,b]}\cup C|$ vertices of $S_{\ipar([a,b])}$ are marked.

In the second phase we mark all the lowest common ancestors of marked vertices,
that is, the vertices $s$ such that in the first phase,
the vertices from at least two distinct subtrees rooted
at children of $s$, were marked.
The common ancestors can be marked in linear time, using post-order
traversal ---
we only need to store for each vertex $s$, whether any element
of the subtree rooted at $s$ was marked in the first phase.
Additionally, we mark the root of every tree with
at least one marked vertex.

Let us count the vertices marked after the second phase.
Remove the subtrees with no marked vertices.
Let $q$ be the number of marked vertices of degree $2$.
If we replace all the degree 2 vertices with edges, we obtain a forest,
where every vertex that is neither the leaf nor the root, has
at least $2$ children.
Denote by $l$ the number of leaves in this forest.
Clearly, it has at most $2l$ vertices.
However, every leaf could be marked only in the first phase
and hence $l+q\leq 2|E_{[a,b]}\cup C|$, so $2l+q\leq 4|E_{[a,b]}\cup C|$.

\begin{corollary}
  After the second phase of the reduction, at most $4|E_{[a,b]}\cup C|$ vertices
  of $S_{\ipar([a,b])}$ are marked.
\end{corollary}

The third reduction phase removes the subtrees with no previously
marked vertices.
All the 2-edge-components represented by vertices from those components
look exactly the same in $G_{[a,b]}$ as well as in all the individual
versions $G_a,\ldots,G_b$ and thus need not be tracked
in the descendants of $T_{[a,b]}$.

In the last phase we replace every remaining path of degree
2 unmarked vertices with a single path vertex.
These vertices may include both simple and path vertices.
However, neither of them has been marked, so for each $x\in[a,b]$,
the underlying path of 2-edge-connected components $c_1,\ldots,c_g$
either remains unaltered in $G_x$ or is a part of a single,
larger 2-edge-component in $G_x$.

Since the number of such paths does not exceed the number
of vertices marked so far, we end up with a forest $S'$
of at most $8|E_{[a,b]}\cup C|$ vertices.

Each phase of the reduction can be implemented as a simple graph search,
so the reduction takes time $O(|V(S_{\ipar([a,b])})|+|E_{[a,b]}|+|C|)$.

After the reduction comes the contraction.
We extend the forest $S'$ with the edges
$E_{[a,b]}$ alive in each of $G_a,\ldots,G_b$.
We merge the 2-edge-connected components found in this graph
into new, simple vertices, obtaining a new graph~$S''$,
which is again a forest.
It may happen that some vertices of $S'$ have not been merged
into larger components in $S''$.
Every such vertex $s\in V(S')$ is a path vertex in $S''$ iff
it is a path vertex in $S'$.
The roots of trees of $S''$ are chosen arbitrarily,
but keeping in mind that the trees should not be rooted
at path vertices.
The properly rooted $S''$ forms our graph $S_P$.
Contraction can be implemented to work in time
$O(|V(S_{\ipar([a,b])})|+|E_{[a,b]}|)$.

Let us bound the time needed to compute $S_{[a,b]}$.
If $b-a+1\geq n$, then the reduction is skipped and thus
the time spent on building $S_{[a,b]}$ is $O(n+|E_{[a,b]}|)$.
Otherwise, the reduction is performed and we spend
$O(|V(S_{\ipar([a,b])})|+|C|+|E_{[a,b]}|)$ time.

The asymptotic running time of building $S_{[a,b]}$
turns out to be exactly the same as in the case of connectivity.
Thus, building a data structure $T$ for representing 2-edge-connectivity
takes the same time.

\begin{corollary}
We can build a data structure $T$ representing 2-edge-connectivity in
a graph timeline in $O(m+t\log{n})$ time.
The space usage is $O(t\log{n})$.
\end{corollary}

The optimization allowing the evaluation of $\compid$ in time
$O(\log{n})$ applies here as well.

\section{Answering $\fa$ queries}\label{sec:forall}
In this section we show how to extend the data structure $T$,
so that it can be used for answering $\fa$ queries.
The preprocessing for $\fa$ queries constitutes another
phase, that we apply only after we computed the
data structure $T$.

Let us begin with a simple observation.
Assume that we want to answer a $\fa(u, w, a, b)$ query, where
$[a,b]$ is an elementary interval.
Then, if the same vertex of $S_{[a,b]}$ represents both $u$ and $w$,
then there is actually a path between $u$ and $w$ in $G_{[a,b]}$ and
we can immediately give a positive answer.
However, the reverse relation is not true.
It may happen that $u$ and $w$ are represented by distinct vertices in $S_{[a,b]}$,
but are connected in each of $G_a, \ldots, G_b$.
Thus, our first goal in this section is to compute, for each two vertices
in each of $S_{[a,b]}$, whether the vertices represented by them are connected
in each of $G_a, \ldots, G_b$.

For an elementary interval $[a,b]$, let
$c_{[a,b]}(s,x)$, where $s\in V(S_{[a,b]})$, $x\in[a,b]$, be the
result of the call $\compid(s,a,b,x)$.
Our goal is to compute for each vertex $s\in S_{[a,b]}$ a \emph{fingerprint},
that is, an integer $H_{[a,b]}(s)\in [1,|V(S_{[a,b]})|]$ with the following property:
the sequences 
$c_{[a,b]}(s,a)c_{[a,b]}(s,a+1)\ldots c_{[a,b]}(s,b)$ and $c_{[a,b]}(s',a)c_{[a,b]}(s',a+1)\ldots c_{[a,b]}(s',b)$ are equal
iff $H_{[a,b]}(s)=H_{[a,b]}(s')$.

To answer a $\fa(u, v, a, b)$ query, where $[a,b]$ is an elementary interval, we first map $u$ and $v$ into vertices $u'$ and $v'$ of $S_{[a,b]}$ and then report a positive answer iff $H_{[a,b]}(u') = H_{[a,b]}(v')$.
In order to handle arbitrary intervals, we decompose the query interval into $O(\log t)$ elementary intervals.
The decomposition as well as the mapping can be implemented as
a function $\faint(s_1,s_2,x,y,a,b)$, whose pseudocode is given in Appendix~\ref{a:code}.
To answer a $\fa(u,v,x,y)$ query we execute $\faint(\lar_{[0,\infty]}[v],\lar_{[0,\infty]}[w],x,y,1,t)$.

Let us now describe the computation of fingerprints.
They are computed in a bottom-up fashion, starting from the leaves of $T$.

\begin{lemma}\label{l:hash}
  Let $P=[a,b]$ be an elementary interval and $s\in V(S_{P})$. Define:
  $$\tilde{H}_{P}(s)=\begin{cases}
  (s, 0) & \mbox{if } \lar_{P}(s)=\perp \textrm{ or } \rar_{P}(s)=\perp \\
    (H_{\lson(P)}(\lar_{P}[s]),H_{\rson(P)}(\rar_{P}[s])) & \mbox{otherwise}.
\end{cases}$$
Then $c_P(s_1,a)\ldots c_P(s_1,b)=c_P(s_2,a)\ldots c_P(s_2,b)$ iff $\tilde{H}_P(s_1)=\tilde{H}_P(s_2)$.
\end{lemma}
\begin{proof}
  If $\lar_P[s_1]=\perp$, then $c_P(s_1,a)=(s_1,P)$ and for each $s_2\in V(S_P)$
  such that $s_1\neq s_2$, we have $c_P(s_2,a)\neq (s_1,P)$.
  The pair $\tilde{H}_P(s_1)=(s_1,0)$ is unique among the pairs $\tilde{H}_P(s)$,
  so we have $c_P(s_1,a)\ldots c_P(s_1,b)=c_P(s_2,a)\ldots c_P(s_2,b)$
  iff $s_1=s_2$.
  Analogously, if $\rar_P[s_1]=\perp$, then $c_P(s_1,b)\neq c_P(s_2,b)$ for
  $s_1\neq s_2$ and thus $s_1$ is given a unique pair $\tilde{H}_P(s_1)=(s_1,0)$.
  Therefore, if $\lar_P[s_1]=\perp$ or $\rar_P[s_1]=\perp$, then $\tilde{H}_P(s_1)=\tilde{H}_P(s_2)$
  is equivalent to $s_1=s_2$.

  It remains to consider the case, when $s_1$ and $s_2$ are represented
  in both $S_{\lson(P)}$ and $S_{\rson(P)}$.
  Let $m=\lfloor (a+b)/2\rfloor$.
  By the definition of the pair $\tilde{H}_{P}(s)$ we have  $c_P(s_i,a)\ldots c_P(s_i,m)=$
  $c_{\lson(P)}(\lar_P[s_i],a)\ldots c_{\lson(P)}(\lar_P[s_i],m)$
  as well as
  $c_P(s_i,m+1)\ldots c_P(s_i,b)=c_{\rson(P)}(\rar_P[s_i],m+1)\ldots c_{\rson(P)}(\rar_P[s_i],b)$.
  The sequences $c_P(s_1,a)\ldots c_P(s_1,b)$ and $c_P(s_2,a)\ldots c_P(s_2,b)$,
  are equal exactly when their corresponding halves are equal, that
  is, by the definition of fingerprints $H$, iff
  $H_{\lson(P)}(\lar_P[s_1])=H_{\lson(P)}(\lar_P[s_2])$
  and $H_{\rson(P)}(\rar_P[s_1])=H_{\rson(P)}(\rar_P[s_2])$.
\maybeqed\end{proof}

Observe that the pairs $\tilde{H}_P(s)$ from the above lemma
satisfy the desired properties of fingerprints,
with the exception that they are pairs of integers, not integers.
Thus, in order to compute the values $H_P(s)$, it suffices to map the values of $\tilde{H}_P(s)$
into distinct positive integers (two pairs are assigned the same integer iff they are equal).
As both numbers in each pair $\tilde{H}_P(s)$ are at most $O(|V{(S_P)}|)$ we may compute the mapping
in linear time by using radix-sort algorithm.
Note that this resembles the Karp-Miller-Rosenberg \cite{Karp:1972} algorithm.
The total additional time and space used is $O(\sum_P |V{(S_P)}|) = O(t\log n)$.
Thus, we obtain an $\langle O(m+t\log n), O(\log t)\rangle$ data structure for
answering $\fa$ queries.

However, the query time can be made independent of the length of the timeline and speeded up to $O(\log n)$.
In order to do that, we employ a shortcutting technique similar to the one
used for finding connected components of vertices in individual graphs
combined with a data structure for comparing the
subwords of a given word.

Assume again that $D$ is the smallest integer such that $2^D\geq n$.
Observe that the $\scut$ table from Section~\ref{sec:tree_structure}
allows us to speed up $\fa(u,v,x,y)$, where
$[x,y]\subseteq [k\cdot2^D+1, (k+1)\cdot 2^D]$ for
some $k\in[0,2^{B-D})$.
Let $\scut[u][k]=(s_u,P_u)$ and $\scut[v][k]=(s_v,P_v)$.
We only need to perform the following steps:
\begin{enumerate}
\item If $P_u\neq P_v$, then the answer is $\textbf{false}$.
\item If $P_u=P_v$, but the length of $P_u$ is more than $2^D$, then the answer is \textbf{true} 
  if and only if $s_u=s_v$.
\item Otherwise, $P_u=P_v = [k\cdot 2^D+1, (k+1)\cdot2^D]$ and the answer can be obtained by calling $\faint(s_u,s_v,x,y,k\cdot2^D+1,(k+1)\cdot2^D)$.
\end{enumerate}
Only the third steps takes superconstant time, namely $O(\log 2^D)=O(\log n)$.

Consider the general query $\fa(u,v,x,y)$, where
$[x,y]\not\subseteq [k\cdot2^D+1, k\cdot (2^D+1)]$, for any $k$.
Let $l_1$ be the smallest integer such that $x< l_1\cdot2^D+1$ and $l_2$
be the largest integer such that $l_2\cdot 2^D<y$.
Then, our query can be split into the conjunction of three
queries: $\fa(u,v,x,l_1\cdot 2^D)$, $\fa(u,v,l_2\cdot 2^D+1,y)$
and $\fa(u,v,l_1\cdot2^D+1,l_2\cdot2^D)$ (we assume the
last query to be \textbf{true} if $l_1=l_2$).
The first two can be answered in $O(\log n)$ time, as discussed
above.
We deal with the third one in a different way.
Assume $l_1<l_2$.
For any~$l$, let $\scut[u][l]=(s_u,P_u)$ and $\scut[v][l]=(s_v,P_v)$.
Define $h_u^l:=(H_{P_u}(s_u),P_u)$.
Observe that the answer to $\fa(u,v,l\cdot 2^D+1,(l+1)\cdot 2^D)$
is affirmative exactly iff $h_u^l = h_v^l$, which follows immediately from the definition of fingerprints.
Thus, to answer $\fa(u,v,l_1\cdot2^D+1,l_2\cdot2^D)$,
we need two check if the sequences $h_u^{l_1}h_u^{l_1+1}\ldots h_u^{l_2-1}$
and $h_v^{l_1}h_v^{l_1+1}\ldots h_v^{l_2-1}$ are equal.
We use the following algorithm, described for instance in \cite{Crochemore:2007}.
It uses the linear construction of a suffix array and the optimal range minimum query structure.

\begin{lemma}[\cite{Crochemore:2007}]\label{l:text}
There exists a data structure that, after a linear preprocessing of
the word $W$, allows us to check in time $O(1)$ if two subwords of $W$ are
equal.
\end{lemma}

Let $X_v=h_v^0h_v^1\ldots h_v^{2^{B-D}-1}$ and let
$X$ be the concatenation $X_1X_2\ldots X_n$.
Notice that the length of $X$ is $O(t)$.
We build the data structure of Lemma~\ref{l:text} for the sequence $X$ in $O(t)$ time.
Hence, we can answer the query $\fa(u,v,l_1\cdot2^D+1,l_2\cdot2^D)$
by comparing the appropriate two subwords of $X$ of length $l_2-l_1$
in time $O(1)$.
\begin{theorem}
There exists an $\langle O(m+t\log n),O(\log n)\rangle$ data structure for answering $\fa$ queries.
\end{theorem}

\subsection{2-edge-connectivity}
As in the case of the connectivity relation, for each $s\in V(S_{[a,b]})$,
we want to encode the entire history of what happens with
$s$ in each of the individual versions $G_a,\ldots,G_b$.
Since we introduced path vertices in the graphs $S_{[a,b]}$,
the appropriate fingerprints need to be defined in a more subtle way.

We partition the vertices of a graph $S_{[a,b]}$ into three groups:
\begin{enumerate}
  \item simple vertices,
  \item \emph{vanishing} path vertices. If $s$ is a vanishing path
    vertex and it represents a path of 2-edge-connected components
    $c_1,\ldots,c_k$, then for each $x\in [a,b]$,
    all the components $c_1,\ldots,c_k$ are parts of a single,
    larger 2-edge-connected component $C_x$ in $G_x$,
  \item \emph{non-vanishing} path vertices. $s\in V(S_{[a,b]})$
    is a non-vanishing path vertex if there exists $x\in [a,b]$ such
    that the underlying 2-edge-connected components $c_1,\ldots,c_k$ are
    all actual 2-edge-connected components of $G_x$.
\end{enumerate}

Let us assume that $P=[a,b]$ and $s\in V(S_P)$.
Denote by $c_P(s,x)$ the result of\linebreak $\compid(s,a,b,x)$.
As in the case of connectivity we define $H_P(s) \in \{1,\ldots,|V(S_P)|\}\cup\{\perp\}$
to be the fingerprint of the sequence $c_P(s,a)\ldots c_P(s,b)$.
We set $H_P(s)=\perp$ only if $s$ is a non-vanishing
path vertex.
Otherwise, if $s$ is simple or a vanishing path vertex, $H_P(s)$
is an integer.

There is a reason why a non-vanishing path vertex is not assigned
an integer fingerprint: if vertices $v, w\in V$, $v\neq w$ are both represented by 
a non-vanishing path vertex $s$ in $S_P$,
then the sequences $c_{[1,t]}(\lar_{[0,\infty]}[v],a)\ldots c_{[1,t]}(\lar_{[0,\infty]}[v],b)$
and $c_{[1,t]}(\lar_{[0,\infty]}[w],a)\ldots c_{[1,t]}(\lar_{[0,\infty]}[w],b)$
might be different.

In order to define the fingerprints $H_P(s)$ based on the fingerprints
for the children intervals $\lson(P)$ and $\rson(P)$,
we first define the initial fingerprints $\tilde{H}_P(s)$.
The values $\tilde{H}_P(s)$ have the same properties as
the fingerprints $H_P(s)$ defined above, except that if
$\tilde{H}_P(s)\neq\perp$, then $\tilde{H}_P(s)$ is a pair
of integers from the range $[1,|V(S_P)|]$.
The initial fingerprints can be computed according to the following rules, which
implicitly decide whether a path vertex $s$ is vanishing or non-vanishing.
\begin{enumerate}
\item If $s$ is a simple vertex:
  \begin{itemize}
    \item if $\lar_P[s]=\perp$ or $\rar_P[s]=\perp$, then $\tilde{H}_P(s)=(s,0)$.
      The fingerprint has to be unique in this case.
    \item Let $s_l=\lar_P[s]$ and $s_r=\rar_P[s]$. If $H_{\lson(P)}[s_l]=\perp$
      or $H_{\rson(P)}[s_r]=\perp$, then $\tilde{H}_P(s)=(s,0)$.
      It is a consequence of $s$ representing exactly a single 2-edge-connected
      component of $G_x$, for some $x\in[a,b]$.
    \item Otherwise, $\tilde{H}_P(s)=(H_{\lson(P)}(s_l),H_{\rson(P)}(s_r))$.
  \end{itemize}
\item If $s$ is a path vertex:
  \begin{itemize}
    \item if $\lar_P[s]=\perp$ or $\rar_P[s]=\perp$, then $\tilde{H}_P(s)=\perp$,
    \item Let $s_l=\lar_P[s]$ and $s_r=\rar_P[s]$. If $H_{\lson(P)}[s_l]=\perp$
      or $H_{\rson(P)}[s_r]=\perp$, then $\tilde{H}_P(s)=\perp$.
    \item Otherwise, $\tilde{H}_P(s)=(H_{\lson(P)}(s_l),H_{\rson(P)}(s_r))$.
  \end{itemize}
\end{enumerate}

Again, we can use radix-sort to convert each value $\tilde{H}_P(s)$ (distinct from $\perp$)
to an integer $H_P(s)$ in the range $[1,|V(S_P)|]$.

We now sketch how to answer the query $\fea(u,v,x,y)$ in time $O(\log{t})$.
As in Section~\ref{sec:forall}, we reduce this problem to answering $O(\log{t})$ queries
with the time period being an elementary interval.
Let us focus on a single elementary interval $P$.
For $w\in\{u,v\}$, denote by $Q_w$ the last interval on the path $[1,t]\to P$ such that
vertex $w$ is represented in $S_{Q_w}$ by a vertex $s_w$.
Also, we denote by $Q_w'$ the last interval on path $[1,t]\to P$ such that
vertex $w$ is represented in $S_{Q_w'}$ by a \emph{simple} vertex $s_w'$.
Denote the quadruple $(Q_w,s_w,Q_w',s_w')$ by $\phi_P(w)$.

We have the following cases:
\begin{enumerate}
  \item $Q_u\neq Q_v$. The answer is clearly \textbf{false}.
  \item $Q_u=Q_v$, $Q_u\neq P$. As $s_u$ or $s_v$ could potentially be path vertices,
    the answer is positive if and only if $(Q_u',s_u')=(Q_v',s_v')$.
  \item $Q_u=Q_v=P$.
    \begin{enumerate}
      \item $H_P(s_u)\neq\perp$ and $H_P(s_v)\neq\perp$. The answer is \textbf{true}
        iff $H_P(s_u)=H_P(s_v)$.
      \item $H_P(s_u)=\perp$ or $H_P(s_v)=\perp$. The answer is the same as the result
        of the comparison $(Q_u',s_u')=(Q_v',s_v')$.
    \end{enumerate}
\end{enumerate}
Once we have the fingerprints, all the above checks can be performed in $O(1)$ time,
therefore we can answer $\fea$ queries in $O(\log{t})$ time.

In order to optimize the query time to $O(\log{n})$,
we adapt the technique used to speed up $\fa$ queries
for connectivity.
Assume again that $D$ is the smallest integer such that $2^D\geq n$.
First we show that we can answer the query
$\fea(u,v,x,y)$, where
$[x,y]\subseteq [k\cdot2^D+1, (k+1)\cdot 2^D]$ for
some $k\in[0,2^{B-D})$,
in time $O(\log{n})$.
We precompute the values $\phi_P(w)$ for
each $w\in V$ and an elementary interval $P$ not longer
than~$2^D$.
As the number of elementary intervals not longer than $2^D$ is
$O(t/n)$, we precompute $O(t)$ values in total.
$\phi_P(w)$ can be computed based on $\phi_{\ipar(P)}(w)$
in constant time, so we spend $O(t)$ time on
precomputation.

Let $[x,y]\subseteq [k\cdot 2^D+1, (k+1)\cdot 2^D]$
and let $P_1,\ldots,P_p$ be the partition of $[x,y]$
into elementary intervals.
As each $P_i$ is a descendant of $[k\cdot 2^D+1, (k+1)\cdot 2^D]$
in the segment tree,
the value $\phi_{[k\cdot 2^D+1, (k+1)\cdot 2^D]}(w)$
can be used as a starting point to compute the values $\phi_{P_i}(w)$.
We need to descend only $B-D$ levels down the tree to compute
the values $\phi_{P_i}(w)$.
Thus, the time needed to answer $\fea(u,v,x,y)$ is $O(B-D)=O(\log{n})$
in this case.

To handle the general query $\fea(u,v,x,y)$, let $l_1$ be the smallest
integer such that $x<l_1\cdot 2^D+1$ and $l_2$ be the largest integer
for which $l_2\cdot 2^D<y$ holds.
Our query can be split into the conjunction of three
queries: $\fea(u,v,x,l_1\cdot 2^D)$, $\fea(u,v,l_2\cdot 2^D+1,y)$
and $\fea(u,v,l_1\cdot2^D+1,l_2\cdot2^D)$ (we assume the
last query to be \textbf{true} if $l_1=l_2$).
The first two can be answered in $O(\log n)$ time, as discussed
above.
In order to answer the last query, we need to rephrase the check 
$\fea(u,v,l\cdot 2^D+1,(l+1)\cdot 2^D)$ in terms of symbol equality $h_u^l=h_v^l$.
Indeed, for $P_l=[l\cdot 2^D+1,(l+1)\cdot 2^D]$ and $\phi_{P_l}(w)=(Q_w,s_w,Q_w',s_w')$
we can set:
$$h_w^l=\begin{cases}
  (Q_w,Q_w',s_w') & \mbox{if } Q_w\neq P_l \\

  (Q_w,H_{P_l}(s_w)) & \mbox{if } Q_w = P_l\mbox{ and }H_{P_l}(s_w)\neq\perp \\
  (Q_w,Q_w',s_w') & \mbox{if } Q_w=P_l\mbox{ and }H_{P_l}(s_w)=\perp.
\end{cases}$$
It can be easily verified that $h_u^l=h_v^l$ if and only if the previously described
checks for answering $\fea(u,v,x,y)$ where $[x,y]=[l\cdot 2^D+1,(l+1)\cdot 2^D]$,
produce a positive answer.

We answer the query $\fea(u,v,l_1\cdot 2^D+1,l_2\cdot 2^D)$ by checking if
the words $h_u^{l_1}\ldots h_u^{l_2-1}$ and $h_v^{l_1}\ldots h_v^{l_2-1}$
are equal.
The needed words are all subwords of the word $X_1\ldots X_n$, where
$X_v=h_v^0h_v^1\ldots h_v^{2^{B-D}-1}$.
The word length is $O(t)$.
By Lemma~\ref{l:text}, after additional preprocessing in $O(t)$ time,
we can answer the query $\fea(u,v,l_1\cdot 2^D+1,l_2\cdot 2^D)$
in constant time.

\begin{theorem}
There exists an $\langle O(m+t\log n), O(\log n) \rangle$ data structure for answering $\fea$ queries.
\end{theorem}

\section{Improved lower and upper bounds for $\ex$ queries}\label{sec:exists}
In this section we focus on $\ex$ queries.
We first give improved conditional lower bounds for answering these queries,
and then show an algorithm, whose running time matches one of the new bounds.
As shown in \cite{Lacki:2013}, the problem of multiplying two Boolean
$n\times n$ matrices can be reduced to the problem of answering
$\Theta(n^2)$ $\ex$ queries about a graph timeline $G^t$, where $t=\Theta(n^2)$.
Denote by $O(n^{\omega'})$ the time required to perform
$n\times n$ Boolean matrix multiplication (BMM).
Thus, unless $\omega'=2$, it is not possible to develop a data structure, 
which after almost linear preprocessing answers $\ex$ queries in polylogarithmic time.
In this section we give several new lower bounds.

Throughout this section, we repeatedly use $\epsilon$
to denote an arbitrarily small, positive number.
The exact value of $\epsilon$ may vary and depend
on the context.
We also denote by $\delta(\epsilon)$ some other small
positive number, dependent on $\epsilon$.

Let us recall the somewhat informal, yet important,
partition of algorithms into \emph{algebraic} and \emph{combinatorial}.
The combinatorial algorithms do not make use of the fact that
the matrices are defined over a ring, i.e., they do not use subtraction.
No $O(n^{3-\epsilon})$ combinatorial algorithm is known for BMM.

We show a connection between the $\ex$ data
structure and algorithmic problems related to detecting triangles in graphs.
In the \emph{triangle detection} problem we are given a graph $G=(V,E)$, where $|E|=m$,
and the goal is to find three vertices $a,b,c\in V$ such that $(a,b),(a,c),(b,c)\in E$.
The best known known algorithm for triangle detection was given
by Alon et al. \cite{Alon:1997} and works in
$O(m^{1.41})$ time.
The best combinatorial algorithm is folklore and runs
in $O(m\sqrt{m})$ time.
The following relation between triangle detection and BMM was shown in \cite{Williams:2010}:
\begin{lemma}\label{l:matrix}
An $O(m^{1.5-\epsilon})$ combinatorial algorithm for
triangle detection implies an $O(n^{3-\delta(\epsilon)})$
combinatorial algorithm for BMM.
\end{lemma}
The related problem is \emph{triangle listing}, where we
are asked to find $c$ triangles in a graph with $m$ edges.
Pătraşcu \cite{Patrascu:2010} proved the following lemma.
\begin{lemma}\label{l:3sum}
If one can list $m$ triangles from a graph with $m$ edges
in $O(m^{4/3-\epsilon})$ time, then there exists an
$O(n^{2-\delta(\epsilon)})$ algorithm for 3-SUM.
\end{lemma}

We now show a relation between triangle listing and $\ex$ queries.

\begin{lemma}\label{l:report}
The problem of listing $c$ triangles in a graph with $m$ edges can be reduced to
answering
$O(m+c\log n)$ $\ex$ queries in a timeline $G^t$ of length $t=O(m)$ and no permanent edges.
\end{lemma}
\begin{proof}
Let $H$ be the input graph, in which we are supposed to list triangles.
Moreover, let $V(H) = \{v_1, \ldots, v_n\}$.
We build a timeline $G^t$ of graphs on vertex set $V(H)$ by processing vertices $v_1, \ldots, v_n$ one by one.
First, we add an empty graph to $G^t$.
Then, for a vertex $v_i$, we append $2\deg_H(v_i)$ new versions to $G^t$ ($\deg_H(v)$ denotes the degree of vertex $v$ in $H$), which we call a \emph{block} of vertex $v_i$.
Within each block, we first create $\deg_H(v_i)$ new versions, at each step adding one more edge incident to $v_i$.
The edges are added in arbitrary order.
Then, we create $\deg_H(v_i)$ more versions by removing the edges incident to $v_i$.
Note that the last graph in every block is empty, and in the middle graph the vertex $v_i$ has degree $\deg_H(v_i)$.
Let the the block of a vertex $v_i$ start at $G_{a_i}$ and end at $G_{b_i}$.

Observe that we obtain a timeline $G^t$, where $t=4m+1$, as each edge of~$H$ is added and removed exactly twice.
For each edge $(v_i,v_j)\in E(G)$, $i < j$, we can test if
there is a triangle $(v_i,v_j,v_k)$, where $j<k$, with a single
query $\ex(v_i,v_j,b_j+1,t)$.
Indeed, the answer to such a query is positive iff there exists
$v_k$ such that there is a path from $v_i$ to $v_j$ in $G_{a_k+deg_H(v_k)-1}$.
The path, along with the edge $(v_i, v_j)$, forms a triangle.

Note that the query $\ex(v_i,v_j,a_p,b_q)$, for $j<p\leq q$, tells us
if there is any triangle $(v_i,v_j,v_k)$ such that $k\in[p,q]$.
Thus, we may use a divide-and-conquer approach for listing triangles,
which is based on the following observation.
If we are looking for triangles such that $k\in [p,q]$,
a negative answer to an $\ex(v_i,v_j,a_p,b_{(p+q)/2})$ query
allows us to halve the search interval.
Hence, we can find all $l$ vertices $v_k$ such that $(v_i,v_j,v_k)$
is a triangle in time $O(l\log{n})$.
The detailed procedure $\reptr$ is given in Appendix~\ref{a:code}.
\maybeqed\end{proof}

By combining Lemmas~\ref{l:matrix}, \ref{l:3sum} and~\ref{l:report}, we obtain the following.

\begin{theorem}
  Let $\Psi$ be a problem of answering $\Theta(t)$ $\ex$ queries about an arbitrary graph timeline $G^t$
  with no permanent edges.
\begin{itemize}
  \item  An $O(t^{1.4})$ algorithm for $\Psi$ implies
    an $O(t^{1.4})$ algorithm for triangle finding.
  \item  An $O(t^{1.5-\epsilon})$ combinatorial algorithm for $\Psi$ implies
    an $O(n^{3-\delta(\epsilon)})$ combinatorial algorithm for BMM.
  \item An $O(t^{4/3-\epsilon})$ algorithm for $\Psi$ implies
    an $O(n^{2-\delta(\epsilon)})$ algorithm for 3-SUM.
\end{itemize}
\end{theorem}

In addition, we show that an $\ex$ data structure with preprocessing/query time product of $O(t^{2-\epsilon})$
and queries substantially faster than $O(\sqrt{t})$ implies
a faster BMM algorithm.
\begin{lemma}\label{l:trdf}
Suppose there exists an $\langle O(t^{2-q-\epsilon}), O(t^q)\rangle$
combinatorial data structure for answering $\ex$ queries, where $q\in[0,\frac{1}{2})$ is a parameter.
Then there exists an $O(n^{3-\delta(\epsilon)})$ combinatorial
algorithm for BMM.
\end{lemma}
\begin{proof}
We use the assumed algorithm for answering $\ex$ queries to develop an $O(n^{3-\delta(\epsilon)})$ combinatorial algorithm for BMM.
By Lemma \ref{l:matrix}, in order to obtain such an algorithm, it suffices to show an $O(m^{\frac32 - \delta(\epsilon)})$ algorithm for triangle detection.
This in turn, by Lemma~\ref{l:report}, can be reduced to answering $O(m)$ $\ex$ queries in a graph timeline $G^t$ with no permanent edges, where $t = O(m)$.
To complete the proof, we show how to answer these queries in $O(m^{\frac32 - \delta(\epsilon)})$ total time.

Split $G^t$ into blocks of length $t^\alpha$ consisting of consecutive versions, where $\alpha\in[0,1]$
is to be set later.
Let $G_a,\ldots,G_b$ be one of the blocks.
First, we compute in time $O(t)$ the edges alive during
entire block and contract the components formed by those edges.
Next, we mark the components subject to any updates in the interval
$[a,b]$.
There are $O(t^\alpha)$ marked components.
All the queries $\ex(u,v,c,d)$ such that $[c,d]\subseteq [a,b]$
and concerning the unmarked components, can be answered
in constant time --- if suffices to check whether components
of $u$ and $v$ are the same.
For the remaining $O(t^\alpha)$ components we build the assumed
data structures.
The initialization takes time $O(t^{\alpha(2-q-\epsilon)}+t)$ per
individual block, which gives $O(t^{1+\alpha(1-q-\epsilon)}+t^{2-\alpha})$ time for $O(t^{1-\alpha})$ blocks.

We can answer an $\ex(u,v,x,y)$ query by going through at most $O(t^{1-\alpha})$
blocks intersecting $[x,y]$ and querying each block structure
once.
Thus, the query overhead is $O(t^{1-\alpha}\cdot t^{\alpha q})=O(t^{1-\alpha(1-q)})$.

By setting $\alpha=\frac{1}{2-2q-\epsilon}$, we have $\alpha\in (\frac{1}{2},1]$.
The exponent of the initialization time becomes less than $\frac{3}{2}$, as we have:
$$1+\alpha(1-q-\epsilon)=1+\frac{1-q-\epsilon}{2-2q-\epsilon}<1+\frac{1-q-\epsilon}{2-2q-2\epsilon}=\frac{3}{2},$$
$$2-\alpha < 2 -\frac{1}{2}=\frac{3}{2},$$
whereas the exponent of time needed to answer $O(t)$ queries is
$$1+1-\alpha(1-q)=2-\frac{1-q}{2-2q-\epsilon}<2-\frac{1-q}{2-2q}=\frac{3}{2}.$$
The lemma follows.
\maybeqed\end{proof}

What is interesting, we can give a combinatorial data structure, whose running time matches the above lower bound.

\begin{theorem}\label{l:trdf_structure}
  For every $0\leq \alpha < 1$ there exists an $\langle O(m+\min(nt,t^{2-\alpha})),O(t^\alpha)\rangle$
data structure for answering $\ex$ queries.
It uses $O(\min(nt,t^{2-\alpha}))$ space.
\end{theorem}
\begin{proof}
  If $n=o(t^{1-\alpha})$, then $nt=o(t^{2-\alpha})$ and the
  $\langle O(m+nt), O(1)\rangle$ data structure for answering
  $\ex$ queries from \cite{Lacki:2013} is sufficient to finish
  the proof.

  Let us assume that $n=\Omega(t^{1-\alpha})$.
We first build the tree-like structure $T$ in $O(m+t\log{n})$ time.
Let $D$ be the largest integer such that $2^D\leq t^\alpha$.
We split the graph timeline into blocks of size $2^{B-D}$, which is roughly $t^{1-\alpha}$.
Denote the versions of the $i$-th block by $G_{a_i},\ldots,G_{b_i}$.
Observe that $[a_i,b_i]$ is an elementary interval.

Denote by $l_i(v)$ the pair $(s,P)$, where $P$ is
the last interval on a path from the root $[1,t]$ to $[a_i,b_i]$
such that $v\in V$ is represented by $s$ in $S_P$.

Consider answering the \emph{block query} $\ex(u,v,x,y)$,
where $[x,y]\subseteq [a_i,b_i]$.
Set $l_i(v)=(s_v,P_v)$ and $l_i(w)=(s_w,P_w)$.
If $P_v\neq P_w$, then the answer is clearly $\textbf{false}$.
Otherwise, if $P_v\neq [a_i,b_i]$, then neither $v$ nor $w$ are
represented in $S_{[a_i,b_i]}$ and thus the answer is
the same as the result of a comparison $s_v=s_w$.
In the last case, when $P_v=P_w=[a_i,b_i]$, the result can vary.
However, we can use the $\langle O(m+nt), O(1)\rangle$ data structure 
for answering $\ex$ queries from \cite{Lacki:2013}.
For each $i$, we build this structure for a graph with vertices $V(S_{[a_i,b_i]})$ and the timeline
induced by the subsequent edge updates in graphs $G_{a_i},\ldots,G_{b_i}$.
The size of $V(S_{[a_i,b_i]})$ as well as the length of this timeline
is $O(b_i-a_i)=O(t^{1-\alpha})$, so it can be initialized in time $O((t^{1-\alpha})^2)$
to answer queries in $O(1)$ time.
In a single query, it gives us the desired answer to $\ex(s_v,s_w,x,y)$.

As a result, once we have found $l_i(v)$ and $l_i(w)$,
we can answer the block query in constant time.
Observe that a general $\ex$ query can be answered by issuing $O(t^{\alpha})$ block queries.
All the required $O(t^\alpha)$ $l_j(u)$ values can be computed
by traversing the first $D$ levels of $T$,
which can be implemented to work in time $O(t^\alpha)$.
We need to build $O(t^\alpha)$ $\ex$ data structures --- each in time $O((t^{1-\alpha})^2)$,
so the total initialization time is $O(m+t^{2-\alpha}+t\log{n})=O(m+t^{2-\alpha})$.
\end{proof}

\section{Open problems}\label{sec:open_problems}
For $\fa$ and $\fea$ queries, we gave an $\langle O(m+t\log n),O(\log n)\rangle$
data structure.
What about the biconnectivity?
Although it is possible to propose a similar tree-like structure
that represents biconnectivity in individual versions, it seems hard
to extend it to $\fea$-like queries.
The main obstacle is biconnectivity relation on vertices not
being an equivalence relation.

It would be also interesting to know whether even faster query
(without sacrificing $O(t\log n)$ initialization time) is possible
for $\fa$ queries.

Concerning $\ex$ queries, we proved that beating our
trade-off structure in the domain of combinatorial algorithms
implies a faster combinatorial matrix multiplication algorithm.
However, is there a way to employ fast matrix multiplication
to obtain a data structure for $\ex$ queries with
preprocessing/query time product of $O(t^{2-\epsilon})$?
\bibliographystyle{plain}
\bibliography{paper}

\renewcommand{\thesection}{\Alph{section}}
\appendix

\section{Omitted pseudocode} \label{a:code}
\begin{algorithm}[H]
\begin{algorithmic}[1]
    \Function{Comp-Id}{$w,a,b,c$} \Comment{$w \in V(S_{[a,b]}), c\in [a,b]$}
        \If{$a=b$}
            \State \Return $(w,[a,b])$
        \EndIf
        \State $\imid:=\left\lfloor\frac{a+b}{2}\right\rfloor$
        \If{$c \leq \imid$}
            \If{$\lar_{[a,b]}[w] = \perp$}
                \State \Return $(w,[a,b])$
            \Else
                \State \Return \Call{Comp-Id}{$\lar_{[a,b]}[w],a,\imid,c$}
            \EndIf
        \Else
            \If{$\rar_{[a,b]}[w] = \perp$}
                \State \Return $(w,[a,b])$
            \Else
                \State \Return \Call{Comp-Id}{$\rar_{[a,b]}[w],\imid+1,b,c$}
            \EndIf
        \EndIf
    \EndFunction

\end{algorithmic}
\end{algorithm}

\begin{algorithm}[H]
\begin{algorithmic}[1]
    \Function{Forall-Aux}{$s_1,s_2,x,y,a,b$}
      \If{$s_1=s_2$}
        \State \Return \textbf{true}
      \EndIf
      \If{$[x,y]=[a,b]$} \Comment{$[x,y]$ is elementary, we refer to fingerprints}
      \If{$H_{[a,b]}(s_1)=H_{[a,b]}(s_2)$}
          \State \Return \textbf{true}
        \Else
          \State \Return \textbf{false}
        \EndIf
      \EndIf
      \State $\imid:=\lfloor\frac{a+b}{2}\rfloor$
      \If{$x\leq \imid$}
        \If{$\lar_{[a,b]}[s_1]=\perp \textbf{ or } \lar_{[a,b]}[s_2]=\perp$}
          \State \Return \textbf{false}
        \EndIf
        \If{\textbf{not }\Call{Forall-Aux}
          {$\lar_{[a,b]}[s_1],\lar_{[a,b]}[s_2],x,\imid,a,\imid$}}
          \State \Return \textbf{false}
        \EndIf
      \EndIf
      \If{$y> \imid$}
        \If{$\rar_{[a,b]}[s_1]=\perp \textbf{ or } \rar_{[a,b]}[s_2]=\perp$}
          \State \Return \textbf{false}
        \EndIf
        \If{\textbf{not }\Call{Forall-Aux}
          {$\rar_{[a,b]}[s_1],\rar_{[a,b]}[s_2],\imid+1,y,\imid+1,b$}}
          \State \Return \textbf{false}
        \EndIf
      \EndIf
      \State \Return \textbf{true}
    \EndFunction
\end{algorithmic}
\end{algorithm}

\begin{algorithm}[H]
\begin{algorithmic}[1]
    \Procedure{Report-Triangles}{G}
      \State construct the timeline $G^t$ from the proof of Lemma~\ref{l:report}
      along with numbers $a_i,b_i$
        \For{$(u,v)\in E(G)$}   \Comment $u<v$
      \State \Call{report-internal}{$u,v,v+1,|V(G)|$}
        \EndFor
    \EndProcedure
    \State
    \Procedure{Report-Internal}{$u,v,x,y$}
        \If{$x>y$}
            \State \Return
        \EndIf
      \If{\textbf{not }$\ex(u,v,a_x,b_y)$}
            \State \Return
        \EndIf
        \If{$x=y$}
            \State output triangle $(u,v,x)$
            \State if $k$ triangles have been reported so far, stop
            \State \Return
        \EndIf
        \State $\imid := \left\lfloor \frac{x+y}{2}\right\rfloor$
        \State\Call{Report-Internal}{$u,v,x,\imid$}
        \State\Call{Report-Internal}{$u,v,\imid+1,y$}
    \EndProcedure
\end{algorithmic}
\end{algorithm}

\end{document}